\documentclass[%
 reprint,
 amsmath,amssymb,
 aps,
]{revtex4-1}

\usepackage{graphicx}
\usepackage{dcolumn}
\usepackage{bm}
\usepackage{amssymb}

\newtheorem{theorem}{Theorem}
\newtheorem{proof}{Proof}
\newtheorem{corollary}{Corollary}
\newtheorem{definition}{Definition}
\newtheorem{lemma}{Lemma}
\newtheorem{lproof}{Proof}
\newtheorem{cproof}{Proof}
\newtheorem{claim}{Claim}
\newtheorem{clproof}{Proof}
\begin{document}


\title{The Non-Null and Force-Free Electromagnetic Field}

\author{Govind Menon}
\affiliation{Department of Chemistry and Physics\\ Troy University, Troy, Al 36082}

\date{\today}

\begin{abstract}
In this paper, we present a covariant formalism that connects solutions to force-free electrodynamics in the non-null case and foliations of spacetime. In doing so, we are also able to derive an expression of the general non-null current density vector. Just as in the null case, solutions in the non-null case can give rise to a dual solution, however, as is shown below, this can happen only when the solution describes a vacuum field. All theorems are illustrated with previously known solutions.
\end{abstract}

\pacs{Valid PACS appear here}
\maketitle


   \section{Introduction}
    Force-free electrodynamics (FFE) describes a set of nonlinear equations wherein the current density vector belongs to the kernel of the Maxwell field tensor $F$. In the magnetosphere of black holes where the plasma density is significantly lower than the electromagnetic field strength density, such conditions are expected to be satisfied. Following the work of Goldreich and Julian (\cite{GJ69}) on pulsars, Blandford and Znajek (\cite{BZ77}) extended their analysis to rotating black holes in what has become a leading mechanism for energy and angular momentum extraction. In the Blandford-Znajek model, the magnetosphere is force-free, axis-symmetric, and stationary.
  
    Subsequently, progress into FFE developed steadily. On the one hand, numerical solutions dominated the scene. For examples of relatively recent efforts see \cite{K04}, \cite{Qian_2018}, \cite{Koide_2019}, and \cite{TT14}. A theoretical understanding of FFE was also pioneered by Uchida (\cite{Uchida1}, \cite{Uchida2}) and Komissarov (\cite{K04}), just to name a few. Near horizon and extreme Kerr magnetospheres have been studied systematically by \cite{Cam_20}, \cite{cam20FF}, \cite{Armas_20}, \cite{GLS_16} and \cite{CO_16}.
    
    The first exact analytical solution to the Blandford-Znajek equations appeared in \cite{MD07}. Here the solution described a null electromagnetic field. Soon thereafter, using both the infalling and outgoing null solutions, a non-null solution was given in \cite{Minout}. By adjusting a single parameter it was possible to make this field either electrically or magnetically dominated. Finally, another magnetically dominated solution was developed in \cite{MenonTetrad15}. Both the numerical and analytical solutions, while insightful,  remained disconnected without a cohesive theoretical language until recently, when  Gralla and Jacobson published the current status of the theory of FFE in 2014 (\cite{GT14}).
 
 In \cite{CGL16}, the authors argue that in the stationary and axis-symmetric case in a Kerr background, magnetically dominated solutions are fully described by two-dimensional Lorentzian foliations of the background Kerr spacetime. In \cite{Menon_FF20}, using an adapted frame formalism, it was shown that regardless of the value of $F^2$, all force-free solutions, in an arbitrary but electrically neutral background spacetime, are determined by the existence of well-prescribed foliations of spacetime. The arguments did not rely on solutions being stationery and axis-symmetric. Here, the causal character of the foliations corresponded to the sign of $F^2$. Despite its generality, the only drawback of the adapted chart formalism was that the covariance of the results was not manifest, nor where the geometric meaning of the constraints on the leaves of the foliation. But, this result provided the first step needed to establish the connection between FFE and foliations.
 
 In a recent paper (\cite{Menon_FFN20}), the results of \cite{Menon_FF20} was completely rewritten in geometric terms for the case of the null and force-free field. Here, it was shown that the leaves of the foliation contain a unique null pregeodesic, and the properties of the null mean curvature of the null congruence determined whether the foliation would allow a null and force-free solution; but when it does, the solutions came as a class of solutions with exactly two free parameters.
 
 In this paper, we will focus on the non-null and force-free electromagnetic field. The primary focus is to rewrite the results of \cite{Menon_FF20} in a completely geometric form so that covariance of the theory of non-null FFE and its associated foliations of spacetime is manifest. After a brief introduction to the equations of FFE, we show that the existence of solutions in the magnetically dominated case depends on the existence of 2-dimensional Lorentzian foliations of spacetime such that its mean curvature 1-form is closed by the dual mean curvature field that arises from the normal distribution of the foliation. The mean curvature here is the trace of the second fundamental form, and closed-ness is in the sense of exterior calculus. In the remainder of the section, we will formulate the analogous results for the electrically dominated case. The differences here are nominal.
 
This paper is written in the vein of \cite{Menon_FF20} and \cite{Menon_FFN20}. Its sole purpose is to develop a deeper understanding of the theory of FFE. The results obtained are not necessarily to device a computational recipe for the construction of new solutions. In this spirit, for illustrative purposes, we include a recasting of the previously known non-null solutions to FFE in a Kerr background using the new language of this paper. However, at this moment, save a new vacuum solution in a Kerr background, we are unable to generate any new solutions with a current density in a Kerr background. We finally conclude his discussion by presenting a few important topics of future study stemming from our analysis.
\section{The Basic Equations Force-Free Electromagnetic field}
In general relativity, spacetime is a 4-dimensional smooth manifold ${\cal M}$ endowed with a metric $g$ of Lorentz signature $(-1, 1,1,1)$. 
In this paper, the metric is predetermined and satisfies the Einstein equation with possibly electrically neutral matter and field content. We single out the requirement for electrical neutrality because we want to account for all electromagnetic effects. 
The electromagnetic field tensor $F$ satisfies
\begin{equation}
 d  F = 0\;,
 \label{Fclosed}
\end{equation}
 and
\begin{equation}
* \;d * F = j\;.
\label{inhomMaxform}
\end{equation}
Here $*$ is the Hodge-Star operator and $d$ is the exterior derivative on forms. Also, $j$ denotes the current density dual vector.
The 3-current $J \equiv d*F$. Then  $*J=j$.
Force-free electrodynamics is defined by the constraint 
$$ F(j^\sharp, \chi) =0$$
for all contravariant vector fields $\chi$.
Here, for any 1-form $w$,
$$w^\sharp = g^{\mu\nu} \;w_\nu \;\partial_\mu\;,$$
while for any tangent vector field $\chi$,
$$\chi^\flat \equiv g_{\mu\nu} \chi^\nu dx^\mu\;.$$
The above expressions are valid in any local chart. 
The Maxwell Field tensor
$F$ is said to be magnetically dominated whenever $F^2 >0$, $F$ is  electrically dominated whenever $F^2 <0$, and finally a  force-free electromagnetic field $F$ is  null whenever $F^2 =0$. 
\vskip0.2in
In the remainder of this section, we will summarize the essential properties of FFE developed over a period of time since it came of prominence following the seminal work of Blandford and Znajek \cite{BZ77}. While some of the relevant previous results can be found in  \cite{Uchida1}, \cite{Uchida2} and \cite{Carter79}, the recent paper by Gralla and Jacobson \cite{GT14} explains all the essential equations of FFE listed below. 
\vskip0.2in
The kernel of $F$, denoted by $\ker F$, is a 2-dimensional involutive distribution of the tangent bundle satisfying the property that $i_v F =0$ whenever $v \in \ker F $. By involutive we mean that, whenever $v,w \in \ker F$, we have that $[v,w] \in \ker F$. Frobenius' theorem then implies that when a force-free $F$ exists on ${\cal M}$, spacetime can be foliated by 2-dimensional integral submanifolds of the distribution spanned by $\ker F$. The leaves of the foliation, which are the integral submanifolds of $\ker F$,  will be denoted as $ {\cal F}_a$. It is usual for the submanifolds $ {\cal F}_a$ to be referred to as a {\it field sheet}. Here $a$ belongs to some indexing set $A$. The key points here are that
$$ {\cal F}_a \cap {\cal F}_b =0 \;{\rm whenever}\;a \neq b \in A\;,\;\;\cup_{a\in A} \;{\cal F}_a ={\cal M}\;,$$
and whenever $v\in T({\cal F}_a) $ for any $a \in A$ we have that $v \in \ker F$. Additionally, a force-free electromagnetic field is a simple 2-form given by
\begin{equation}
    F= \alpha \wedge \beta\;,
    \label{simF}
\end{equation}
for some 1-forms $\alpha$ and $\beta$. The force-free condition can be reduced to
\begin{equation}
    J \wedge \alpha = 0 =J \wedge \beta\;.
    \label{FFwedgecond}
\end{equation}

\section{The Geometry of the Non-Null Force-Free Field}
\subsection{The Magnetically Dominated Force-Free Field}
In eq.(\ref{simF}), since $\alpha$ and $\beta$ span a 2-dimensional plane, we can always pick them to be orthogonal to each other and still obtain the correct expression for $F$. Therefore, without loss of generality, we set $g(\alpha, \beta)=0$. Then in the magnetically dominated case
$$F^2= 2 \alpha^2 \beta^2 >0\;,$$ where $\alpha^2 = g(\alpha, \alpha) = g^{\mu\nu}\; \alpha_\mu \alpha_\nu$, and similarly for $\beta$. I.e., the 2-dimensional plane spanned by $\alpha$ and $\beta$ is spacelike. Consequently, the spacetime metric when restricted to $\ker F$ consisting of all vectors anhilated by $\alpha$ and $\beta$ is a 2-dimensional plane of Lorentz signature. Then about any point in spacetime we can construct and inertial frame field  $(e_0, e_1, e_2, e_3)$ such that
$$g(e_\mu, e_\nu) = \eta_{\mu\nu}\;,$$
for $\mu = 0, 1, 2,3$ and where $\eta$ is the Minkowski metric
\begin{equation}
   \eta= \left(
    \begin{array}{cccc}
      -1 & 0 &0&0 \\
       0 & 1 &0&0 \\
      0 & 0&1&0 \\
      0&0&0&1
    \end{array}
  \right) \;. 
  \label{eta}
\end{equation}
$e_0$ and $e_1$ span $\ker F$ and forms an involutive distribution. 
Further, $F$ can now be written as
\begin{equation}
    F= u\;e_2 ^\flat \wedge e_3 ^\flat\;.
    \label{FrameF}
\end{equation}
In the above expression, $u$ is a yet to be determined component function for $F$. 
\vskip0.2in
Conversely, consider any foliation ${\cal F} = \{{\cal F}_a: a \in A\}$ of spacetime by 2-dimensional Lorentzian manifolds ${\cal F}_a$. By this, we mean that the metric when restricted to ${\cal F}_a$ has a Lorentz signature. We will refer to such a foliation as a 2-D {\it Lorentzian foliation} of ${\cal M}$ and  denote it by ${\cal F}^{\;2L}$. The leaves of a 2-D Lorentzian foliation will be denoted as ${\cal F}_a ^{\;2L}$. For any $p \in {\cal M}$ there exists an open set $U_p$ of $p$ and an inertial frame  $(e_0, e_1, e_2, e_3)$ on $U_p$ with the properties listed above such that the tangent space for each of the leaves of the foliation when restricted to $U_p$ is spanned by $e_0$ and $e_1$. Roughly speaking, the tangent bundle of the foliation is locally given by,
\vskip0.2in
$T({\cal F}_a ^{\;2L})|_{U_p}$

$=  \big\{(x, v_x)  \; | \; x \in U_p \cap {\cal F}_a ^{\;2L}, v_x \in {\rm span}\;\;\{e_0(x), e_1(x) \}\big\}\;.$
\vskip0.2in \noindent
We will refer to such frames as 2-D {\it  Lorentzian foliation adapted frames} and denote it by ${\bf F}_{2L}$.
Since by construction, the distribution spanned by $e_0$ and $e_1$ is integrable, from Frobenius' theorem (for details see \cite{JLee13}) we get that
\begin{equation}
    de_2 ^\flat= e_2^\flat \wedge A + e_3 ^\flat \wedge B\;,
    \label{de2flat}
\end{equation}
and
\begin{equation}
     de_3 ^\flat= e_2^\flat \wedge C + e_3 ^\flat \wedge D\;,
    \label{de3flat}
\end{equation}
for some 1-forms $A,B,C$ and $D$.

We already know that a magnetically dominated force-free electromagnetic field gives rise to a ${\cal F} ^{\;2L}$. The interesting question is, given a ${\cal F} ^{\;2L}$, what conditions must be satisfied for there to be an associated magnetically dominated, force-free electromagnetic field? As we shall see, the key players in the theory of magnetically dominated FFE are the associated {\it mean curvature field} $H$ of leaves of the foliation ${\cal F}_a ^{\;2L}$, and its dual field which we denote by $\tilde H$.
\vskip0.2in
Let $V, W$ be vector fields tangent to any ${\cal F}_a ^{\;2L}$. Then the {\it shape tensor} or {\it  second fundamental form} $\Pi$ of ${\cal F}_a ^{\;2L}$ is defined by
$$\Pi (V, W) = (\nabla_V W)^\perp\;.$$
Here $\perp$ takes the component of the vector normal to the surface ${\cal F}_a ^{\;2L}$. The mean curvature field at any point of ${\cal F}_a ^{\;2L}$ is then defined by
$$ H=\frac{1}{2} \Big[- \Pi(e_0, e_0) + \Pi(e_1, e_1)\Big]$$
in any ${\bf F}_{2L}$. It is easy to see that $H$ is independent of the frame used since it is the metrically contracted version of $\Pi$, i.e., if $h$ is the induced metric on ${\cal F}_a ^{\;2L}$, then
$$H =\frac{1}{2}\; h^{\mu\nu}\; \Pi (e_\mu, e_\nu)\;.$$
Additionally, even though the complimentary orthogonal distribution of the foliation spanned locally by $e_2$ and $e_3$ may not be an involutive distribution, we may still define a ``dual" mean curvature field by
$$\tilde H=\frac{1}{2} \Big[ \Pi(e_2, e_2) + \Pi(e_3, e_3)\Big]\;. $$
When written in ${\bf F}_{2L}$, $H$ and $\tilde H$ take the explicit form:
\vskip0.2in
$2H= \left[-g(\nabla_{e_0} e_0, e_2)+ g(\nabla_{e_1} e_1, e_2) \right] e_2$
\begin{equation}
  + \left[-g(\nabla_{e_0} e_0, e_3)+ g(\nabla_{e_1} e_1, e_3) \right] e_3\;, 
    \label{Hdef}
\end{equation}
and
\vskip0.2in
$2\tilde H= \left[-g(\nabla_{e_2} e_2, e_0)- g(\nabla_{e_3} e_3, e_0) \right] e_0$
\begin{equation}
  + \left[g(\nabla_{e_2} e_2, e_1)+ g(\nabla_{e_3} e_3, e_1) \right] e_1\;.  
   \label{Htildedef}
\end{equation}
Before we state the central result of this section, we begin with a lemma.
\vskip0.2in
\begin{lemma}
$$de_i ^\flat (e_i, e_j) = g(\nabla_{e_i} e_i, e_j)\;,$$
for
$$i=0,1 \;\;{\rm and}\;\; j=2,3\;,$$
and also when
$$i=2,3 \;\;{\rm and}\;\; j=0,1\;.$$
\end{lemma}
\begin{lproof}
$$de_0 ^\flat (e_0, e_3) = (e_0 ^\mu e_3 ^\nu - e_3 ^\mu e_0 ^\nu)\; \nabla_\mu (e_0)_\nu$$
$$= e_3 ^\nu\; \nabla_{e_0} (e_0)_\nu - e_0 ^\nu\; \nabla_{e_3} (e_0)_\nu = e_3 ^\nu\; \nabla_{e_0} (e_0)_\nu\;.$$
The other expressions can be shown in a similar way.
$\blacksquare$
\end{lproof}
We are now able to state the main theorem of this paper as follows:
\vskip0.2in
\begin{theorem}
Let ${\cal F}^{\;2L}$ be any $2$-D Lorentzian foliation of ${\cal M}$ with leaves $\{{\cal F}_a ^{\;2L}, \;a \in A\}$. Let ${\bf F}_{2L}$ be a Lorentzian frame field on a starlike open set $U_p$ about any $p \in {\cal M}$. Then, up to a constant factor in $u$, $F$ given by eq.(\ref{FrameF}) is a unique magnetically dominated force-free electrodynamic field in $U_p$ such that
$$\ker F|_{U_p} = \cup_a \;T({\cal F}_a ^{\;2L} \cap U_a)$$
if and only if
\begin{equation}
    dH^\flat=-d\tilde H^\flat\;,
    \label{H+tilH}
\end{equation}
where $H (/\tilde H)$ are the mean (/dual) curvature field associated with the foliation. Moreover, in this case, 
\begin{equation}
 d (\ln u) = 2 (H+\tilde H)^\flat\;. 
 \label{eqforu}
\end{equation}
\label{Mdomthm}
\end{theorem}
\begin{proof}
Let $(e_0, e_1, e_2, e_3)$ be an adapted frame ${\bf F}_{2L}$ in $U_p$. Then eq.(\ref{simF}) implies that, tentatively,
\begin{equation}
    F= u\; e_2 ^\flat \wedge e_3^\flat\;.
    \label{magF}
\end{equation}
Noting that the metric here is simply the Minkowski metric (eq.(\ref{eta})), we get that 
\begin{equation}
    *F = u \; e_0 ^\flat \wedge e_1^\flat\;.
    \label{starF}
\end{equation}
Then eq.(\ref {FFwedgecond}) requires that
$$0=(d*F) \wedge e_2 ^\flat\;,$$
which gives
$$e_3 (\ln u)=-de_0 ^\flat (e_0, e_3) + de_1 ^\flat (e_1, e_3)\;.$$
From the previous lemma, this becomes
$$e_3 (\ln u)=-g(\nabla_{e_0} e_0, e_3)+g(\nabla_{e_1} e_1, e_3)\;,$$
or
\begin{equation}
  d(\ln u) (e_3)= 2 H^\flat (e_3)\;.  
  \label{magdomeq1}
\end{equation}
Similarly
$$0=(d*F) \wedge e_3 ^\flat\;,$$
gives that
$$e_2 (\ln u)=-de_0 ^\flat (e_0, e_2) + de_1 ^\flat (e_1, e_2)\;.$$
Just as above, this becomes
\begin{equation}
  d (\ln u) (e_2)= 2 H^\flat (e_2)\;.  
  \label{magdomeq2}
\end{equation}
Finally, enforcing the homogeneous Maxwell equation $dF=0$, we get that
$$\Big(e_1(\ln u)+de_2 ^\flat (e_1, e_2) + de_3 ^\flat (e_1, e_3)\Big) e_1 ^\flat \wedge e_2 ^\flat \wedge e_3 ^\flat =0$$
and
$$\Big(e_0(\ln u)+de_2 ^\flat (e_0, e_2) + de_3 ^\flat (e_0, e_3)\Big) e_0 ^\flat \wedge e_2 ^\flat \wedge e_3 ^\flat =0\;,$$
which using the previous lemma reduces to
\begin{equation}
  d(\ln u) (e_1)= 2 \tilde H^\flat (e_1)\;.  
  \label{magdomeq3}
\end{equation}
and
\begin{equation}
  d(\ln u) (e_0)= 2 \tilde H^\flat (e_0)\;.  
  \label{magdomeq4}
\end{equation}
From eqs.(\ref{de2flat}) and (\ref{de3flat}) the components of $dF$ along
$$e_0 ^\flat \wedge e_1 ^\flat \wedge e_2 ^\flat\;{\rm and}\;\;e_0 ^\flat \wedge e_1 ^\flat \wedge e_3 ^\flat$$
 vanish. Therefore, we have met all the requirements of the force-free condition. Eqs.(\ref{magdomeq1}) through (\ref{magdomeq4}) give eq.(\ref{eqforu}). From the Poincare lemma (see \cite{Menon_FF20}), eq.(\ref{eqforu}) has a unique solution in a starlike neighborhood, up to a constant factor if and only if $d^2 \ln u =0$, and this further implies eq.(\ref{H+tilH}).
$\blacksquare$
\end{proof}
Since, given a foliation  ${\cal F}^{\;2L}$, $H$ and $\tilde H$ are well-defined vector fields, the previous theorem is a geometric result and is impervious to a change in the chosen chart or frame. The following definition gives the criterion that a 2-D Lorentzian foliation corresponds to a field sheet in the magnetically dominated case.
\vskip0.2in
\begin{definition}
A $2$-D Lorentzian foliation of ${\cal M}$ given by ${\cal F}^{\;2L}$ 
is a foliation by  field sheets for a magnetically dominated force-free field if and only if eq.(\ref{H+tilH}) holds.
\end{definition}
\vskip0.2in
\begin{theorem}
Given a field sheet foliation ${\cal F}^{\;2L}$ for a magnetically dominated force-free field, in a foliation adapted frame ${\bf F}_{2L}$, the current density vector generating the field is given by
\begin{equation}
j_M= -u\; de_0 ^\flat(e_2, e_3) \;e_0 + u\; de_1 ^\flat(e_2, e_3)\; e_1\;.
     \label{magj1}
\end{equation}
\label{jmthm1}
\end{theorem}
\begin{proof}
From eq.(\ref{starF}), we see that

$d*F = \Big(e_3 (u)+u\;de_0 ^\flat (e_0, e_3) -u\; de_1 ^\flat (e_1, e_3)\Big)\;e_0 ^\flat \wedge e_1 ^\flat \wedge e_3 ^\flat$

$+\Big(e_2 (u)+u\;de_0 ^\flat (e_0, e_2) -u\; de_1 ^\flat (e_1, e_2)\Big)\;e_0 ^\flat \wedge e_1 ^\flat \wedge e_2 ^\flat$

$+ u\; de_0 ^\flat (e_2, e_3)\;e_1 ^\flat \wedge e_2 ^\flat \wedge e_3 ^\flat - u\; de_1 ^\flat (e_2, e_3)\;e_0 ^\flat \wedge e_2 ^\flat \wedge e_3 ^\flat$.
\vskip0.2in \noindent
Therefore, from eqs. (\ref{magdomeq1}) and (\ref{magdomeq2}) we get that
$$d*F =  u\; de_0 ^\flat (e_2, e_3)\;e_1 ^\flat \wedge e_2 ^\flat \wedge e_3 ^\flat - u\; de_1 ^\flat (e_2, e_3)\;e_0 ^\flat \wedge e_2 ^\flat \wedge e_3 ^\flat\;.$$
Noting that $j = *d*F$, we get the needed result.
$\blacksquare$
\end{proof}
It is possible to write a simpler expression for $j$ that will prove useful.
\begin{theorem}
In a foliation adapted frame ${\bf F}_{2L}$,
\begin{equation}
  j_M= u\; g([e_2, e_3], e_0)\;e_0 - u\; g( [e_2,e_3], e_1 )\; e_1\;. 
  \label{magj2}
\end{equation}
\label{jmthm2}
\end{theorem}
\begin{proof}
$$de_0 ^\flat(e_2, e_3) = e_2 ^\mu e_3 ^\nu \;\big(\nabla_\mu (e_0)_\nu - \nabla_\nu (e_0)_\mu\big)$$
$$= g(\nabla_{e_2} e_0, e_3) - g(\nabla_{e_3} e_0, e_2)$$
$$=- g(\nabla_{e_2} e_3, e_0) + g(\nabla_{e_3} e_2, e_0) = g(\nabla_{e_3} e_2-\nabla_{e_2} e_3, e_0)\;.$$
Since the Levi-Civita connection is torsion free, we finally get that
$$de_0 ^\flat(e_2, e_3) = g(e_0, [e_3, e_2])\;.$$
Similarly
$$de_1 ^\flat(e_2, e_3)= g(e_1, [e_3,e_2])\;.$$
$\blacksquare$
\end{proof}
\vskip0.2in
\begin{corollary}
If the distribution spanned by $e_2$ and $e_3$ is involutive, then the magnetically dominated field given by eq.(\ref{FrameF}) is a vacuum solution.
\end{corollary}
\begin{cproof}
In this case $$g([e_3, e_2], e_0) = 0=g([e_3, e_2], e_1)\;.$$ The result follows from the expression in eq.(\ref{magj2}).
$\blacksquare$
\end{cproof}
The expression for $j_M$ as given in eq.(\ref{magj1}) (or equivalently in eq.(\ref{magj2})) is seemingly dependent on the particular choices of the frame field ${\bf F}_{2L}$.
There are two independent transformations for the choices of the frame field. Consider first a transformation of the type
\begin{equation}
  \begin{pmatrix}
\bar e_2\\
\bar e_3\\
\end{pmatrix}= O(2) \begin{pmatrix} e_2\\
e_3 \\
\end{pmatrix}\;,
\label{23rot}
\end{equation}
where $O(2)$ is a $2 \times 2$ orthogonal matrix. In this case $(e_0, e_1, \bar e_2, \bar e_3)$ would be an equally valid frame field. Additionally, one could also have a transformation of type
\begin{equation}
  \begin{pmatrix}
\bar e_0\\
\bar e_1\\
\end{pmatrix}= \Lambda \begin{pmatrix} e_0\\
e_1 \\
\end{pmatrix}\;,
\label{12rot}
\end{equation}
where $\Lambda $ is a $2 \times 2$ general homogeneous transformation satisfying
$$\eta_2 = \Lambda ^T\; \eta_2 \;\Lambda\;,$$
where $\eta_2$ is the 2-D Lorentzian metric given by
\begin{equation}
   \eta= \left(
    \begin{array}{cccc}
      -1 & 0  \\
       0 & 1  \\
    \end{array}
  \right) \;. 
  \label{eta2}
\end{equation}
Under the above transformation, $(\bar e_0, \bar e_1,  e_2,  e_3)$ would be an equally valid frame field. 
\vskip0.2in
\begin{theorem}
The expression for the current density vector in eq.(\ref{magj1})  does not depend on the chosen frame ${\bf F}_{2L}$.
\end{theorem}
\begin{proof}
We have to show that the form of $j$ looks the same under transformations given by eqs. (\ref{23rot}) and (\ref{12rot}). Consider first a choice of simple boost for $\Lambda$. I.e.,
$$\Lambda = \begin{pmatrix}
\gamma & -\beta \gamma\\
-\beta \gamma & \gamma\\
\end{pmatrix}\;,$$
where, as usual (note, we are happy to include non-orthochronous transformations here)
$$\gamma = \pm \frac{1}{\sqrt{1-\beta^2}}\;.$$
Clearly
$$d \bar e_0 ^\flat (e_2, e_3) = \gamma\; d  e_0 ^\flat (e_2, e_3) - \beta \gamma\; d  e_1 ^\flat (e_2, e_3)\;,$$
and
$$d \bar e_1 ^\flat (e_2, e_3) = -\beta \gamma\; d  e_0 ^\flat (e_2, e_3) + \gamma\; d  e_1 ^\flat (e_2, e_3)\;.$$
In this case
$$\frac{\bar j_M}{u} = - \;d\bar e_0 ^\flat(e_2, e_3) \;\bar e_0 + \; d\bar e_1 ^\flat(e_2, e_3)\; \bar e_1$$
$$ = - \Big[ \gamma\; d  e_0 ^\flat (e_2, e_3) - \beta \gamma\; d  e_1 ^\flat (e_2, e_3)\Big] (\gamma e_0-\beta\gamma e_1)$$
$$+ \Big[-\beta \gamma\; d  e_0 ^\flat (e_2, e_3) + \gamma\; d  e_1 ^\flat (e_2, e_3)\Big](-\beta \gamma e_0+\gamma e_1)$$
$$=-\; de_0 ^\flat(e_2, e_3) \;e_0 + \; de_1 ^\flat(e_2, e_3)\; e_1 = \frac{j_M}{u}\;.$$
The expression for $j_M$ is also invariant under the choice $\Lambda = \eta_2$. I.e., $j_M$ is invariant under eq.(\ref{12rot}). In exactly the same way, $j_M$ is invariant under eq.(\ref{23rot}).
$\blacksquare$
\end{proof}
It is easy to see that in the magnetically dominated case, in any ${\bf F}_{2L}$, the energy momentum tensor is given by
$$T = u^2 \left[e_2 \otimes e_2 + e_3 \otimes e_3 - \frac{1}{2}\eta \right]\;.$$

\subsection{The Electrically Dominated Force-Free Field}
Once again picking $\alpha$ and $\beta$ to be orthogonal to each other, since in the electrically dominated case
$$F^2= 2 \alpha^2 \beta^2 <0\;,$$
we get that the 2-dimensional plane spanned by $\alpha$ and $\beta$ has a Lorentzian signature. Consequently, the spacetime metric when restricted to $\ker F$ in the electrically dominated case is spacelike. Then, as before, we can construct a local inertial frame field  $(e_0, e_1, e_2, e_3)$ such that
$e_2$ and $e_3$ span $\ker F$ and forms an involutive distribution. 
Further, $F$ can now be written as
\begin{equation}
    F= \tilde u\;e_0 ^\flat \wedge e_1 ^\flat\;.
\label{FrameFElec}
\end{equation}
\vskip0.2in
For the electrically dominated case, we will denote 2-D Riemannian foliations of ${\cal M}$ by ${\cal F}^{\;2R}$, and the leaves of the foliation will be denoted as ${\cal F}_a ^{\;2R}$. Then, for any $p \in {\cal M}$ there exists an open set $U_p$ of $p$ and an inertial frame  $(e_0, e_1, e_2, e_3)$ on $U_p$ with the properties listed above such that the tangent space for each of the leaves of the foliation when restricted to $U_p$ is spanned by $e_2$ and $e_3$. 
We will refer to such frames as 2-D {\it  Riemannian foliation adapted frames} and denote it by ${\bf F}_{2R}$.
\vskip0.2in
Analogous to the magnetically dominated force-free electromagnetic field, the electrically dominated force-free field gives rise to a ${\cal F} ^{\;2R}$. Given a ${\cal F} ^{\;2R}$, it is now meaningful to ask, what conditions must be satisfied for there to be an associated electrically dominated, force-free electromagnetic field? For continuity of discussion, we do not redefine $H$ and $\tilde H$ as defined in eq.(\ref  {Hdef}) and (\ref  {Htildedef}). However, it is important to remember that this time it is $\tilde H$ that is the mean curvature field of the integral submanifolds of $\ker F$. The result in the electrically dominated case can be obtained by repeating the previous computation in the proof of theorem \ref{Mdomthm} and is similar in appearance. For this reason, we simply state the result without proof.

\begin{theorem}
Let ${\cal F}^{\;2R}$ be any $2$-D Riemannian foliation of ${\cal M}$ with leaves $\{{\cal F}_a ^{\;2R}, \;a \in A\}$. 
Let ${\bf F}_{2R}$ be a Riemannian frame field on a starlike open set $U_p$ about any $p \in {\cal M}$. Then, up to a constant factor in $\tilde u$, $F$ given by eq.(\ref{FrameFElec}) is a unique electrically dominated force-free electrodynamic field in $U_p$ such that
$$\ker F|_{U_p} = \cup_a \;T({\cal F}_a ^{\;2R} \cap U_a)$$
if and only if
\begin{equation}
    dH^\flat=-d\tilde H^\flat\;,
\end{equation}
where $H$ and $\tilde H$ are the mean curvature fields associated with the foliation. Moreover, in this case, 
\begin{equation}
 d (\ln \tilde u) = 2 (H+\tilde H)^\flat\;. 
 \label{eqforutilde}
\end{equation}
\label{Edomthm}
\end{theorem}
The field sheet in the electrically dominated case can now be stated in the following way.
\vskip0.2in 
\begin{definition}
A $2$-D Riemannian foliation of ${\cal M}$ given by ${\cal F}^{\;2R}$ is a foliation by  field sheets for an electrically dominated force-free field if and only if eq.(\ref{H+tilH}) holds.
\end{definition}
By the appearance of eqs.(\ref{eqforu}) and (\ref{eqforutilde}) it might appear that we may obtain a dual solution under a mild restriction as in case of the null and force-free electromagnetic field (see \cite{Menon_FFN20}). This is true here as well, however, as we shall see, this happens only in the vacuum case.
\begin{theorem}
Let ${\cal F}^{\;2R}$ be a field sheet foliation for an electrically dominated force free field on ${\cal M}$. Let $(e_0, e_1, e_2, e_3)$ be a foliation adapted frame ${\bf F}_{2R}$ on some starlike chart on $U_p$ centered about $p \in {\cal M}$ with an associated field $F_E$ as given by eq.(\ref{FrameFElec}). Suppose  $e_0$ and $e_1$ form an involutive distribution, then $F_M$ given by eq.(\ref{FrameF}) is a magnetically dominated solution for $u = \tilde u$ in $U_p$. Moreover, in this case, both solutions describe a vacuum field and up to a negative sign are Hodge star duals of each other.
\label{vacdual}
\end{theorem}
\begin{proof}
Since $u$ and $\tilde u$ satisfy the same equation, all that remains is that $e_0$ and $e_1$ form an involutive distribution. Then theorem \ref{Mdomthm} applies. Moreover
$$* \;e_0 ^\flat \wedge e_1 ^\flat = - e_2 ^\flat \wedge e_3 ^\flat\;.$$
Therefore,
$$* F_E = -F_M\;.$$
If both $F$ and $*F$ are solutions to Maxwell's equations, then they describe vacuum solutions.
$\blacksquare$
\end{proof}
In the same way, magnetically dominated solutions can have dual electrically dominated solutions as well. But, again this is a case of vacuum fields related by the Hodge star operator. The result analogous to theorems \ref{jmthm1} and \ref{jmthm2} in the electrically dominated case is given by the following theorem. Owing to its similarity to the magnetically dominated case, we state the result without proof.
\vskip0.2in
\begin{theorem}
Given a field sheet foliation ${\cal F}^{\;2R}$ for an electrically dominated force free field, in a foliation adapted frame ${\bf F}_{2R}$, the current density vector generating the field is given by 
\begin{equation}
j_E=  \tilde u\; de_2 ^\flat(e_0, e_1) \;e_2 + \tilde u\; de_3 ^\flat(e_0, e_1)\; e_3\;,
     \label{elecj1}
\end{equation}
or equivalently
\begin{equation}
  j_E= -\tilde u\; g( [e_0, e_1], e_2)\;e_2 - \tilde u\; g( [e_0,e_1], e_3 )\; e_3\;. 
  \label{elecj2}
\end{equation}
\end{theorem}

It is also not difficult to see that in the electrically dominated case, in any ${\bf F}_{2R}$, the energy momentum tensor is given by
$$T = \tilde u^2 \left[e_0 \otimes e_0 - e_1 \otimes e_1 + \frac{1}{2}\eta \right]\;.$$
\section{Global Results}
Let $\Lambda_q ({\cal M})$ denote that space of all rank $q$ differential forms on ${\cal M}$. Then the q-th De Rham cohomology group $H^q({\cal M})$ of ${\cal M}$ is defined by the quotient space
$$H^q({\cal M})= \frac{\ker (d: \Lambda_q ({\cal M}) \rightarrow \Lambda_{q+1} ({\cal M}))}{{\rm Im}\; (d:\Lambda_{q-1} ({\cal M}) \rightarrow \Lambda_q({\cal M}))}\;,$$
where ${\rm Im}$ is the image of a map.

\begin{theorem} Let ${\cal F}^{2L}$ be a field sheet foliation for a magnetically dominated force-free field in ${\cal M}$. Let $p \in {\cal M}$. Then if $H^1({\cal M})= \emptyset$, there exists a unique, force-free, magnetically dominated field on all of ${\cal M}$ such that
\begin{equation}
   F^2  (p) =  a^2 
   \label{uniFM}
\end{equation}
for some non-zero constant $a$. 
\end{theorem}
\begin{proof}
Since the function $u$ does not depend on the adapted frame, for a global solution to exist, there must exist a smooth function $u$ such that (eq.(\ref{eqforu}))
$$ d (\ln u) = 2 (H+\tilde H)^\flat\;.$$
But,
$$d (H+\tilde H)^\flat = 0\;.$$
Then $H^1({\cal M})= \emptyset$ implies that there is a global solution for $u$. Now, if $u_1$ and $u_2$ are two such non-trivial solutions, since ${\cal M}$ is connected
$$u_1 = c\; u_2$$ for some non-zero constant $c$. Then, eq.(\ref{uniFM}) fixes $c$ to be $1$.
$\blacksquare$
\end{proof}
For the existence of global non-null force-free solutions, it is not necessary that $H^1({\cal M})= \emptyset$. It is merely a sufficient requirement. 
Owing to its similarity to the previous case, the analogous results for a global electrical dominated field is stated below without proof.

\begin{theorem} Let ${\cal F}^{2R}$ be a field sheet foliation for an electrically dominated force-free field in ${\cal M}$. Let $p \in {\cal M}$. Then if $H^1({\cal M})= \emptyset$, there exists a unique, force-free, electrically dominated field on all of ${\cal M}$ such that
\begin{equation}
   F^2  (p) =  -a^2 
   \label{uniFE}
\end{equation}
for some non-zero constant $a$. 
\end{theorem}
\section{Illustrative Examples}
Exact solutions to force-free electrodynamics are not easily obtained since the governing equations are nonlinear. Moreover, this paper is not intended to necessarily generate new solutions, but rather to promote an understanding of the theoretical aspects of FFE when the solution is non-null. Nonetheless, to illustrate the theory developed above, we revisit and reformulate the only two known non-null force-free class of solutions in a Kerr background. Both solutions will be presented in the Boyer-Lindquist coordinates of the exterior Kerr geometry.
In the Boyer-Lindquist coodinate system $(t,r,\theta,\varphi)$, the Kerr metric takes the form:
$$ds^2 = g_{tt} \;dt^2 + 2 \;g_{t\varphi}\; dt\; d\varphi+ \gamma_{rr}\; dr^2 +
\gamma_{\theta \theta}\; d\theta^2 + \gamma_{\varphi \varphi}\;d\varphi^2,$$
where
$$g_{tt} =  -1 + \frac{2Mr}{\rho^2}, \;\;\; g_{t \varphi}  = \frac{-2Mr a \sin^2\theta}{\rho^2},$$
$$\gamma_{rr} = \frac{\rho^2}{\Delta}\;,\;\;\;\gamma_{\theta \theta} = \rho^2, \;\;\; \gamma_{\varphi \varphi} = \frac{\Sigma^2 \sin^2\theta}{\rho^2}\;.$$
Here
$$ \rho^2 = r^2 + a^2 \cos^2\theta \;,\;\;\Delta = r^2 -2 M r + a^2 \;,$$
and
$$\Sigma^2 = (r^2 + a^2)^2 -\Delta \; a^2 \sin^2\theta.$$
\subsection{Example A}
The original derivation of the class of solutions presented here can be found in \cite{MenonTetrad15}. In \cite{Menon_FFN20} we showed that null foliations of spacetime admitting an equipartition of null mean curvature allows for a class of null force-free solutions. However, in the non-null case this is not true. So you might wonder what we mean by a class of solutions in this example. In this example, we will construct a class of field sheet foliations for a magnetically dominated force-free field, and for each member of the class of foliations, we will have a unique solution. Consider the inertial tetrad in the exterior Kerr geometry in Boyer-Lindquist coordinates given by
\begin{subequations}
\begin{eqnarray}
  e_0 =\frac{1}{\sqrt{\rho^2 \Delta}} \Big[(r^2+a^2) \partial_t + a \partial_\varphi\Big]\;,  \\
  e_1 =\frac{\Big[a \sin^2\theta \partial_t +L \sin\theta \partial_\theta + \partial_\varphi\Big]}{\sqrt{\rho^2 (1+L^2) }\sin\theta} \;,\\
 e_2 = \sqrt{\frac{\Delta}{\rho^2}}\; \partial_r\;,\\
 e_3 =\frac{\Big[L(a \sin^2\theta  \partial_t +  \partial_\varphi)- \sin\theta \;\partial_\theta \Big]}{\sqrt{\rho^2 (1+L^2)}\sin\theta}\;.
\end{eqnarray}
\label{F2L1}
\end{subequations}
Here
$$L = L(r) = \frac{F}{\sqrt{C^2 - F^2}}\;,$$ where $F$ is an arbitrary function of $r$, and $C$ is a constant.
Clearly $(e_0, e_1)$ forms an involutive distribution. They do not however commute. This is one of the advantages of the formalism developed in this paper, as compared to \cite{Menon_FF20}, wherein we require an adapted chart that necessitates a commuting set of basis vector fields.
By construction, $(e_0, e_1, e_2, e_3)$ forms a Lorentzian foliation adapted frame ${\bf F}_{2L}$. We are looking to construct a magnetically dominated force-free field $F$ of the type given in  eq.(\ref{FrameF}). A straightforward calculation shows that
\begin{equation}
    g(\nabla_{e_0} e_0, e_3) =\frac{a^2 \cos \theta \sin \theta}{\rho^2\sqrt{\rho^2 (1+L^2) }} \;,
    \label{term1}
\end{equation}
\begin{equation}
    g(\nabla_{e_0} e_0, e_2) =\frac{a^2 \cos^2 \theta (r-M) + r (Mr-a^2)}{\rho^2\sqrt{\rho^2 \Delta }} \;,
    \label{term2}
\end{equation}
\begin{equation}
    g(\nabla_{e_1} e_1, e_3) =\frac{ \cos \theta (r^2+a^2)}{\rho^2\sqrt{\rho^2 (1+L^2) }\sin \theta} \;,
    \label{term3}
\end{equation}
\begin{equation}
    g(\nabla_{e_2} e_2, e_1) =\frac{ a^2 \cos \theta \sin\theta L}{\rho^2\sqrt{\rho^2 (1+L^2) }} \;,
    \label{term4}
\end{equation}
\begin{equation}
    g(\nabla_{e_3} e_3, e_1) =\frac{ -\cos \theta (r^2+a^2)L}{\rho^2\sqrt{\rho^2 (1+L^2) }\sin \theta} \;,
    \label{term5}
\end{equation}
\begin{equation}
    g(\nabla_{e_1} e_1, e_2) =-\frac{r}{\rho^2}\sqrt{\frac{\Delta}{\rho^2}} \;,
    \label{term6}
\end{equation}
and finally
\begin{equation}
    g(\nabla_{e_2} e_2, e_0) =0= g(\nabla_{e_3} e_3, e_0) \;.
    \label{term7}
\end{equation}
Then from eqs.(\ref{Hdef}), (\ref{Htildedef}),(\ref{eqforu}), and (\ref{term1})-(\ref{term7}) , we get that
$$d (\ln u)  = \frac{ \cot \theta }{\sqrt{\rho^2 (1+L^2) }} (e_3 ^\flat-Le_1 ^\flat ) - \frac{(r-M)}{\sqrt{\rho^2 \Delta }} e_2 ^\flat $$
$$= -\cot \theta \;d\theta-\frac{(r-M)}{\Delta} dr\;.$$
The above equation is easily integrated to obtain a final expression for $u$ given by
$$u =\frac{u_0}{\sin \theta \sqrt{\Delta}}\;,$$
where $u_0$ is an arbitrary integration constant.
When the above equation is inserted into eq.(\ref{FrameF}), and is expressed in terms of the Boyer-Lindquist coordinates we get that the magnetically dominated field in this case is given by
\vskip0.2in \noindent
$F_M=$
\begin{equation}
   \frac{u_0}{\Delta} dr \wedge \left[\frac{\rho^2\sqrt{C^2-F^2}}{\sin \theta} d\theta + F( a \;dt - (r^2+a^2) d\varphi)\right]\;. 
   \label{FM1sol}
\end{equation}
This is precisely the magnetically dominated solution presented in \cite{MenonTetrad15} albeit in a $3+1$ formalism of electrodynamics using electric and magnetic fields.  
As per theorem \ref{vacdual}, there is a dual vacuum solution generated by the above $F_M$ in the event that $e_2$ and $e_3$ forms an involutive distribution. This is the case when $L$ and consequently $F$ is a constant. When $F$ is a constant, the dual vacuum electrically dominated solution is given by
\vskip0.2in \noindent
$F_E= -*F_M=$
$$\frac{u_0}{\sqrt{1+L^2}} dt \wedge d\varphi -\frac{u_0}{\sin \theta}\; d\theta \wedge \left(-dt + a \sin^2\theta d\varphi-\frac{\rho^2}{\Delta} dr\right)\;.$$
It is curious to note that $F_E$ is comprised of three separate vacuum solutions in Kerr geometry. In \cite{Menon_FF20}, it was shown that, 
$$u_0 \frac{ \rho^2}{\sin \theta \Delta}\; d\theta \wedge  dr\;,$$
and its Hodge-star dual, 
$$q\; dt \wedge d\varphi $$
for a constant $q$, are vacuum solutions is Kerr geometry. 
Further
\begin{equation}
   \frac{u_0}{\sin \theta}\; d\theta \wedge (-dt + a \sin^2\theta d\varphi)
   \label{vacnul}
\end{equation}
is the vacuum limit of the outgoing null solution presented below. Since vacuum solutions to electrodynamics are linear even in curved spacetime, we get an additional vacuum solution in Kerr geometry from eq.(\ref{FM1sol}) given by
$$\frac{u_0}{\Delta} dr \wedge ( a \;dt - (r^2+a^2)\; d\varphi)\;.$$
As far as the author is aware, this is the first time the above solution has been presented in the literature.
\subsection{Example B}
The following example has been previously analyzed in \cite{Minout}. The non-null solution presented here is a linear combination of the infalling null solution derived in \cite{MD07} and its outgoing extension in \cite{Minout}. Since FFE is a non-linear theory, linear combinations of solutions are not new FFE solutions in general. However, under certain circumstances, this can indeed be the case. For the case of our example, the precise condition for this to happen is given by eq.(\ref{distcond}). The presentation here is novel and does not follow the path described in the original derivation.

The infalling and outgoing principal null geodesics of the Kerr geometry in Boyer-Lindquist coordinates are given by
$$n = \frac{r^2+a^2}{\Delta} \partial_t -\partial_r + \frac{a}{\Delta} \partial_\varphi$$
and
$$l = \frac{r^2+a^2}{\Delta} \partial_t +\partial_r + \frac{a}{\Delta} \partial_\varphi\;.$$
Let $U(\theta)$ and $V(\theta)$ be smooth functions such that
\begin{equation}
UV = \frac{C}{\sin^2 \theta}\;, 
  \label{distcond}
\end{equation}
where $C$ is a negative constant.
For ease of calculation, set
$$\chi = \frac{1}{2} \sqrt{\frac{\Delta}{\rho^2}} \frac{1}{\sqrt{-U V}}\;.$$
Further, define
\begin{subequations}
\begin{eqnarray}
  e_0 =\chi\; ( U n - V l)\;,  \\
  e_1 =\frac{1}{\sin \theta} \frac{1}{\sqrt{\rho^2}} (a \sin^2 \theta \partial_t + \partial_\varphi)\;,\\
 e_2 = \frac{1}{\sqrt{\rho^2}} \partial_\theta\;,\\
 e_3 =\chi\;( U n + V l)\;.
\end{eqnarray}
\label{2F2L1}
\end{subequations}
Then, eq.(\ref{distcond}) gives us that $[e_0, e_1] = 0\;.$
Therefore $(e_0, e_1, e_2, e_3)$ forms a Lorentzian foliation adapted frame ${\bf F}_{2L}$. 
We are looking to construct a magnetically dominated force-free field $F$ of the type given in  eq.(\ref{FrameF}).
The relevant terms in $H$ and $\tilde H$ are easily calculated, and are given by
\begin{equation}
    g(\nabla_{e_2} e_2, e_1) =0\;,
    \label{2term1}
\end{equation}
\begin{equation}
    g(\nabla_{e_2} e_2, e_0) =\frac{r}{\rho^2}\; \chi\; (U+V)\;,
    \label{2term2}
\end{equation}
\begin{equation}
    g(\nabla_{e_1} e_1, e_2) =-\frac{\cot \theta (r^2+a^2)}{(\rho^2)^{3/2}}\;,
    \label{2term3}
\end{equation}
\begin{equation}
    g(\nabla_{e_1} e_1, e_3) =-\frac{r}{\rho^2}\; \chi\; (V-U)\;,
    \label{2term4}
\end{equation}
\begin{equation}
  g(\nabla_{e_0} e_0, e_2)  = -\frac{a^2 \cos \theta \sin \theta}{(\rho^2)^{3/2}}\;,
  \label{2term5}
\end{equation}
\begin{equation}
  g(\nabla_{e_0} e_0, e_3)  = 2 \chi^2 e_3\left(\frac{UV \rho^2}{\Delta}\right)\;,
  \label{2term6}
\end{equation}
\begin{equation}
  g(\nabla_{e_3} e_3, e_0 ) =- 2 \chi^2 e_0\left(\frac{UV \rho^2}{\Delta}\right) \;,
  \label{2term7}
\end{equation}
and
\begin{equation}
  g(\nabla_{e_3} e_3, e_1 ) = 0\;.
  \label{2term8}
\end{equation}
Eqs. (\ref{2term1})-(\ref{2term8}), eq.(\ref{Hdef}) and eq.(\ref{Htildedef}) now gives us that
$$H + \tilde H =\left[2 \chi^2 e_0\left(\frac{UV \rho^2}{\Delta}\right)-\frac{r \chi}{\rho^2} (U+V)\right] e_0$$
$$+\frac{1}{(\rho^2)^{3/2}} \Big[a^2 \cos \theta \sin \theta - \cot \theta (r^2+a^2)\Big]e_2$$
\begin{equation}
   -\left[2 \chi^2 e_3\left(\frac{UV \rho^2}{\Delta}\right)+\frac{r \chi}{\rho^2} (V-U)\right] e_3\;.
   \label{exzactmagsol}
\end{equation}
Set 
\begin{equation}
   u = 2\sqrt{\frac{-UV}{\Delta}}\;.
   \label{defu}
\end{equation}
\vskip0.2in
\begin{claim} $u$ as given above and $H + \tilde H$ as given in eq.(\ref{exzactmagsol}) satisfies the following required relation
$$d \ln u = (H + \tilde H)^\flat\;.$$
\end{claim}
\begin{clproof}
Clearly
$$d \ln u (e_1) = 0=(H + \tilde H)^\flat (e_1)\;.$$
To check the $e_2$ component of the equation note that
$$\partial_\theta \ln u = \frac{1}{\sqrt{-UV \Delta}} \;\partial_\theta (-UV)$$
$$= -\frac{C}{\sqrt{-UV \Delta}} \; \partial_\theta \csc^2 \theta= -2 \sqrt{\frac{-UV}{\Delta}} \cot \theta = -u \cot \theta\;.$$
It is the equation above that necessitated the requirement in eq.(\ref{distcond}).
Thus we see that
$$e_2 \ln u = - \frac{\cot \theta}{\sqrt{\rho^2}}\;.$$
On the other hand, from eq.(\ref{exzactmagsol})
\vskip0.2in
$(H + \tilde H)^\flat (e_2) =$
$$ \frac{1}{(\rho^2)^{3/2}} \Big[a^2 \cos \theta \sin \theta - \cot \theta (r^2+a^2)\Big]= - \frac{\cot \theta}{\sqrt{\rho^2}}\;. $$
Therefore,
$$d \ln u (e_2) = (H + \tilde H)^\flat (e_2)\;.$$
We will now show the equality in the $e_3$ component of the equation in the claim.
$$e_3 \ln u =  \chi \;(Un + V l) \;\ln \left(2\sqrt{\frac{-UV}{\Delta}}\right)$$
$$= \chi \;(Un + V l) \;\left[\frac{1}{2}\ln \left(\frac{-UV}{\Delta}\right) + 2\right]$$
$$= \frac{\chi \Delta }{2 } \;(Un + V l) \; \left(\frac{1}{\Delta}\right)\;.$$
On the other hand, from eq.(\ref{exzactmagsol})
\vskip0.2in
$(H + \tilde H)^\flat (e_3)$
$$=  -2 \chi^2 \rho^2 e_3\left(\frac{UV}{\Delta}\right)-2 \chi^2 \frac{UV }{\Delta}e_3(\rho^2)-\frac{r \chi}{\rho^2} (V-U) $$
$$= \frac{\chi \Delta }{2 } \;(Un + V l) \; \left(\frac{1}{\Delta}\right)\;.$$
Therefore,
$$d \ln u (e_3) = (H + \tilde H)^\flat (e_3)\;.$$
The $e_0$ component of the equation can be shown in a similar way.
$\blacksquare$
\end{clproof}
Therefore, from theorem \ref{Mdomthm}, the above claim, eqs.(\ref{2F2L1}) and eq.(\ref{FrameF}), we get that
\begin{equation}
  \bar F_M = d\theta \wedge (U n^\flat+V l^\flat) 
  \label{FM}
\end{equation}
is a magnetically dominated force-free field in Kerr geometry. We place a ``bar" over $F$ to distinguish it from the previous solution given by eq.(\ref{FM1sol}). The current density in this case is given by eq.(\ref{magj2}) yields
$$\bar j_M = -\frac{(U \cot\theta + U_{,\theta})}{\rho^2}\; n-\frac{(V \cot\theta + V_{,\theta})}{\rho^2}\; l\;.$$
In the event $C$ in eq.(\ref{distcond}) is positive, we get an analogous electrically dominated solution in a Kerr background, and is only known electrically dominated, non-vacuum, force-free solution in Kerr geometry. Both of these solutions have been previously analysed in \cite{Minout}.

However, when $U = A \csc \theta$ and $V = B \csc \theta$ for constant $A$ and $B$ such that $AB = C$, we get that $e_2$ and $e_3$ in eqs.(\ref{2F2L1}) forms an involutive distribution. In this case $\bar j_M =0$, and once again
$\bar F_E= -*\bar F_M$ is the dual, vacuum, electrically dominated solution. In the event $V =0$, in eq.(\ref{FM}), the formalism developed here does not apply. However, we recover the original  null in-falling solution in the Kerr background, and when $U=0$ we get the null outgoing solution. Further, when $U=0$ and $V= B \csc \theta$, eq.(\ref{FM}) reduces to eq.(\ref{vacnul}).

\section{Conclusion}
In a recent paper (\cite{Menon_FF20}), it was shown that force-free solutions to electrodynamics in an arbitrary spacetime are intimately connected to the existence of prescribed foliations of spacetime. This paper relied on the existence of a foliation adapted chart to establish this connection. The equations constraining the allowed foliations lacked geometric meaning. Although the results were necessarily covariant, it was not apparent from mere inspection. Nonetheless, the connection between FFE and foliations was no longer in doubt, and it was clear that a geometric formulation would be soon forthcoming. It was also clear that null and non-null solutions to FFE belonged to two very different categories of foliations. In \cite{Menon_FFN20}, the theory of null and force-free fields and its connection to foliations of spacetime was formulated in a geometric and covariant manner. This works, then, completes the task of recasting the remaining non-null FFE and its correspondence with foliations in a covariant formalism.

As was shown the main body of the paper, non-null field sheets imposed conditions on the mean curvature field of the foliations and its orthogonal distribution. Consequently, the correspondence between foliations and non-null FEE solutions was a 1-1 map modulo an integration constant. To illustrate the central theorems of this paper, two known solutions in a Kerr background were presented using this covariant formalism. As mentioned previously, this paper focuses on the theoretical structure behind non-null FFE. The formalism developed here can be explored geometrically to look for exact solutions of FFE. However, this is beyond the current scope of work and can be a topic of future study. Additionally, it was shown that modulo an integration constant, given a field sheet foliation, all solutions are locally unique. Existence of global unique solutions as a possibility would appear to be tractable, and this too is a possible topic of future study.

Using a  3+1 formalism, a crucial result by Komissarov was that the magnetically dominated force-free field is indeed governed by a hyperbolic set of equations (\cite{komi2002}). This certainly implies that in a given spacetime, the equations of FFE in the magnetically dominated case implies the existence of field sheet foliations for a magnetically dominated force-free field. However, it is not yet clear how this connection might unfold. This task would be a non-trivial exercise since field sheets in this case will not be contained in a single Cauchy hypersurface. Such a result will involve a significant amount of geometric analysis. With optimism, this too can be relegated as a topic of future study.

\bibliography{bibliography}

\end{document}